\documentclass[sn-mathphys,Numbered]{sn-jnl}

\usepackage{graphicx}%
\usepackage{epstopdf}
\usepackage{multirow}%
\usepackage{amsmath,amssymb,amsfonts}%
\usepackage{amsthm}%
\usepackage{mathrsfs}%
\usepackage[title]{appendix}%
\usepackage{xcolor}%
\usepackage{textcomp}%
\usepackage{manyfoot}%
\usepackage{booktabs}%
\usepackage{algorithm}%
\usepackage{algorithmicx}%
\usepackage{algpseudocode}%
\usepackage{listings}%

\usepackage{dcolumn}
\usepackage{bm}
\usepackage{braket}
\usepackage{dsfont}
\usepackage{ulem}
\usepackage{todonotes}
\usepackage{subfig}

\theoremstyle{thmstyleone}%
\newtheorem{theorem}{Theorem}
\newtheorem{proposition}[theorem]{Proposition}%

\theoremstyle{thmstyletwo}%
\newtheorem{remark}{Remark}%

\theoremstyle{thmstylethree}%
\newtheorem{definition}{Definition}%
\newtheorem{problem}{Problem}%

\newcommand{\Tr}{\text{Tr}}
\newcommand{\tr}{{\rm Tr}}

\begin{document}

\title[Optimizing Measurements Sequences for Quantum State Verification]{Optimizing Measurements Sequences for Quantum State Verification}


\author*[1]{\fnm{Weichao} \sur{Liang}}\email{weichao.liang@yahoo.com}

\author[1]{\fnm{Francesco} \sur{Ticozzi}}\email{ticozzi@dei.unipd.it}

\author[1]{\fnm{Giuseppe} \sur{Vallone}}\email{vallone@dei.unipd.it}

\affil*[1]{\orgdiv{Department of Information Engineering}, \orgname{University of Padova}, \orgaddress{ \city{Padova}, \postcode{35131}, \country{Italy}}}




\abstract{We consider the problem of deciding whether a given state preparation, i.e., a source of quantum states, is accurate, namely produces states close to a target one within a prescribed threshold. We show that, when multiple measurements need to be used, the order of measurements is critical for quickly assessing accuracy. We propose and compare different strategies to compute optimal or suboptimal measurement sequences either relying solely on {\it a priori} information, i.e., the  target state for state preparation, or actively adapting the sequence to the previously obtained measurements. Numerical simulations show that the proposed algorithms reduce significantly the number of measurements needed for verification, and indicate an advantage for the adaptive protocol especially assessing faulty preparations.}




\maketitle

\section{Introduction}
Due to the unavoidable errors, noise or decoherence, realistic quantum devices do not always behave as expected. Various  metrics can be used to characterize and benchmark a quantum device \cite{eisert2020quantum}.
In this work, we focus on devices expected to reliably produce some target state.
Given an unknown quantum state in a $d$-dimensional Hilbert space $\mathcal{H}_d$, $d^2-1$ measurements are necessary in general for a full tomographic reconstruction of the corresponding density matrix~\cite{paris2004quantum}. 

However, in many situations~, such as quantum telecommunication, quantum state preparation, quantum computation etc., we are more concerned with whether some experimentally-accessible quantum state
$\rho_{\rm exp}$ is accurate enough with respect to a target state $\rho_0$, representing the intended result of the preparation, processing or communication task, rather than fully reconstructing it. This problem is referred to as {\it quantum state certification} \cite{certificationtutorial}.

Of course, one way to tackle the problem would be to proceed with a full tomography of $\rho$, and then decide accuracy consequently. This is in general however not efficient as it requires to obtain averages for at least $d^2-1$ independent observables, and it does not leverage on prior information about the target state $\rho_0$. For example, if the target state is known to be pure, a smaller number of measurements are required via compressed sensing techniques \cite{quantumcompressed}.

If the measurements can be designed in an optimized way on the state to be certified, more efficient techniques can be devised, that require less measurements and less repetitions to obtain reliable certification with a specified probability \cite{pallister2018optimal,certificationtutorial}. The basic intuition is that if the state is pure, the state itself is an optimal measurement and repeating the measurement will lead to a quadratic advantage in the number of tests to be performed to achieve a desired accuracy. The strategy can then be extended to include locality constraints, specific classes of target states, adversarial choices in the states to be tested, classical communication and more \cite{pallister2018optimal,wang2019optimal,dangniam2020optimal,yu2019optimal,Jiang2020,li2020optimal,liu2019efficient,liu2020efficient}. A common assumption of these algorithms, to avoid false negatives and ensuring the quadratic advantage, is that all measurements leave the target state invariant.

In this work we reconsider this task, which we shall call the {\it verification problem}, in a different scenario: we assume that only a {\it finite} set of measurement is available and given. Under this assumption, the previously-recalled optimal verification strategies are typically not effective, as it is possible that no measurements leaves the target invariant. 

We thus  construct procedures that decide whether the state $\rho$ is accurate within a prescribed tolerance, without necessarily obtaining a full tomography and thus still reducing the number of required observables. The central idea is to order the measurement sequence using the {\it a priori} information, so that the first measurements are the most informative when the state to be measured is indeed $\rho_0.$ 
The procedure can also be seen as a way to optimize the order of the measured observables in a tomography depending on the best available estimate of the state at hand, in the spirit of \cite{zorzi2013minimum}.  

The procedure we propose are of two types: the first ones compute the whole measurement sequence off-line, and then uses it choose which measurements to actually perform, stopping as soon as verification can be decided. A crucial aspect, in practice, is in fact the computation of the optimal sequence. The latter is a nontrivial optimization problem that has to be solved in a number of instances that scales combinatorially with the number of measurements, which in turn grows at least quadratically in the dimension of the space. For this reason, even off-line calculation becomes of optimal sequences becomes rapidly impractical. In order to address this problem, we propose iterative algorithms, which determine the best next measurements given the previously chosen ones. Two versions are provided, where the second one relies on a relaxation of the constraints that allows for an analytic treatment.
These ways of constructing the sequence, albeit suboptimal, are computationally treatable and offer another advantage: they lend themselves to be used as adaptive strategies, which rely on the previously obtained actual measurements rather than just the target state. In fact,  the second type of verification method we propose is an {\it adaptive} strategy, where the the next measurement is chosen based on the best available estimate given the actual measurement performed to that point. 
The different methods are tested with a paradigmatic example: a two-qubit state where only local Pauli measurements are available.  The results highlight the flexibility of the adaptive method, which performs well even in the case of inaccurate priors.

\section{Problem setting and verification criteria}
We denote by $\mathcal{B}(\mathcal{H})$ the set of all linear operators on a finite dimensional Hilbert space $\mathcal{H}$. 
Define $\mathcal{B}_{*}(\mathcal{H}):=\{X\in\mathcal{B}(\mathcal{H})|X=X^\dag\}$ and $\mathcal{B}_{>0}(\mathcal{H}):=\{X\in\mathcal{B}(\mathcal{H})|X> 0\}$.  $\mathrm{Tr}(A)$ indicates the trace of $A\in\mathcal{B}(\mathcal{H})$.  
We define $\mathcal{S}(\mathcal{H}_d):=\{\rho\in\mathcal{B_*}(\mathcal{H}_d)|\,\rho\geq 0, \mathrm{Tr}(\rho)=1\}$ as  the set of all physical density matrices on $\mathcal{H}_d$. 

In order to precisely specify the verification task, we introduce the following definition, which depends on the choice of a relevant distance-like function.

\begin{definition}[$(\epsilon,D,\rho_0)$-accurate]
Given a target state $\rho_0\in\mathcal{S}(\mathcal{H}_d)$ and a (pseudo-)distance function $D$ on $\mathcal{S}(\mathcal{H}_d)$, the density matrix $\rho\in\mathcal{S}(\mathcal{H}_d)$ is called $(\epsilon,D,\rho_0)$-accurate if $D(\rho,\rho_0)\leq \epsilon$ with $\epsilon\geq 0$.
\end{definition}

Consider a set of observables, represented by Hermitian matrices $\{A_i\}_{i=1}^R$ where $R$ is a positive integer. This set of observables is called information-complete if $\{A_i\}_{i=1}^R$ generate the set of all $d$-dimensional traceless Hermitian matrices. Note that a necessary condition for the observables to be information-complete is $R\geq d^2$.
If $\{A_i\}_{i=1}^R$ is information-complete and the measurement statistics $\{\hat{y}_i\}_{i=1}^R$ are known exactly, i.e., $\hat{y}_i=y_i:=\mathrm{Tr}(\rho_{\rm exp}A_i)$ with $i\in\{1,\dots,R\}$, then there is a unique state compatible with the constraints, that is the generated state $\rho_{\rm exp}$. 
Throught this paper, we suppose that set of observables is {\it finite, information-complete, and fixed.}
The problem we will be concerned with is the following,
\begin{problem}
Based on the \textnormal{a priori} state $\rho_0$ and available data $\{\hat{y_i}\}^{K}_{i=1}$ with $K\leq R$, determine the optimal order of $A_k$ to verify if the generated state $\rho_{\rm exp}$ is $(\epsilon,D,\rho_0)$-accurate via as few measurements as possible.
\end{problem}

In order to introduce the central idea of the work, let us assume for now that a certain sequence of the available observables has been decided.
There are two cases in which the verification process can be terminated, establishing whether the generate state is $(\epsilon,D,\rho_0)$-accurate or not with a minimum of measurements.
Suppose that the measurements are perfect, namely the available data $y_i$ satisfy $y_i=\Tr(A_i \rho_{\rm exp})$. Denote by $\bar{\mathbf{S}}_i:=\{\rho\in\mathcal{S}(\mathcal{H}_d)|\,\mathrm{Tr}(\rho A_i)=y_i\}$ the set of states compatible with the measurement data $y_i$. Based on  $\{y_i\}^{K}_{i=1}$, two criteria can be used to verify if the generated state $\rho_{\rm exp}$ is $(\epsilon,D,\rho_0)$-accurate in each step:\\
\textbf{C1}. If $\min_{\rho\in\bigcap^K_{i=1}\bar{\mathbf{S}}_i} D(\rho,\rho_0)>\epsilon$, $\rho_{\rm exp}$ is not $(\epsilon,D,\rho_0)$-accurate;\\
\textbf{C2}. If $\max_{\rho\in \bigcap^K_{i=1} \bar{\mathbf{S}}_i} D(\rho,\rho_0)\leq \epsilon$, $\rho_{\rm exp}$ is $(\epsilon,D,\rho_0)$-accurate.\\

Depictions of the situations corresponding to the above two  criteria \textbf{C1} and \textbf{C2} are shown in Figure~\ref{Fig:QSV_criteria}. \textbf{C1} guarantees that all states compatible with the measurement data are outside of the ball of radius $\epsilon$ around the target state $\rho_0$, while \textbf{C2} ensures that the same states are all inside.
\begin{figure}[!t]
    \centering
    \subfloat[{\centering {\bf C1}: $\rho_{\rm exp}$ is determined not to be $(\epsilon,D,\rho_0)$-accurate}]{{\includegraphics[height=4cm]{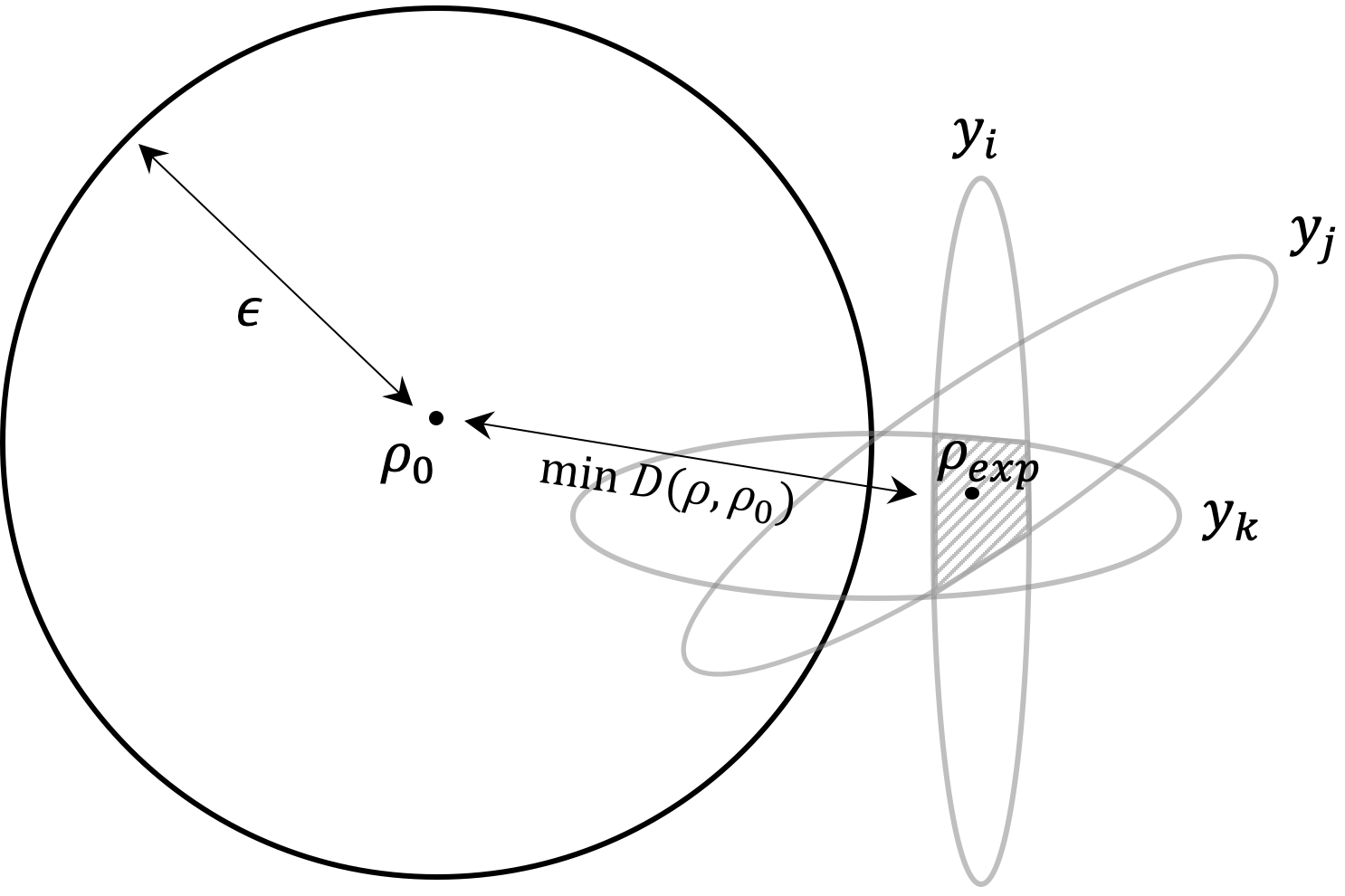} }}%
    \qquad
    \subfloat[{\centering {\bf C2}: $\rho_{\rm exp}$ is determined to be $(\epsilon,D,\rho_0)$-accurate}]{{\includegraphics[height=4cm]{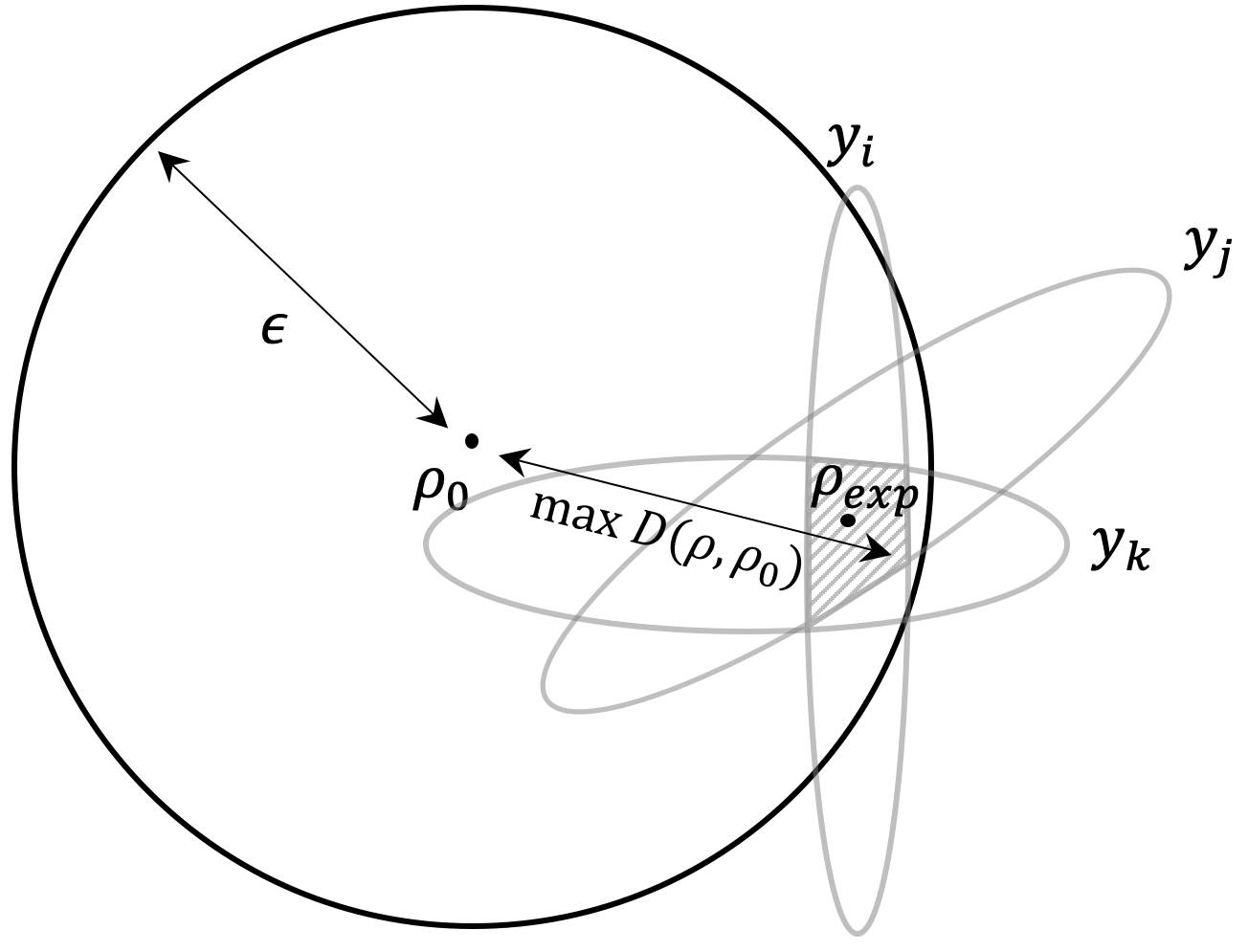} }}%
    \caption{Diagrams corresponding to the quantum state verification criteria \textbf{C1} and \textbf{C2}. The grey area represents $\bar{\mathbf{S}}_i\cap\bar{\mathbf{S}}_j\cap\bar{\mathbf{S}}_k$, i.e., the states compatible with the measurement data $y_i$, $y_j$ and $y_k$. }%
    \label{Fig:QSV_criteria}
\end{figure}

In the following sections, we shall leverage the criteria above in order to devise optimal measurement sequences, or suboptimal ones that present computational advantages and can be adapted to the actual measurement outcomes.

\section{Verification of quantum state based on the \textit{a priori} state}\label{sec3}

In this section, we first introduce a strategy of determining the measurement sequence $\mathsf{M}$ off-line based only on the \textit{a priori} target state $\rho_0$, i.e., without using the measurement data. We next use the sequence $\mathsf{M}$ to verify that the generated states $\rho_{\rm exp}$ is or is not $(\epsilon,D,\rho_0)$ accurate according to the criteria \textbf{C1} and \textbf{C2}. The objective is to perform as few measurements as possible to achieve verification.

\subsection{Off-line construction of the optimal measurement sequence}
From an experimental point of view, it is arguably easier to determine the whole sequence of measurements before performing them. We shall start by exploring this approach, while the adaptive approach, in which the next measurement is chosen depending on the outcome of the previous ones, will be treated in Section \ref{sec:adaptive}. 

Denote by $\mathbf{S}_i(\rho_0):=\{\rho\in\mathcal{S}(\mathcal{H}_d)|\,\mathrm{Tr}(\rho A_i)=\mathrm{Tr}(\rho_0 A_i)\},$ the set of density matrices compatible with the measurement $A_i$ that we would have if the state was actually $\rho_0\in\mathcal{S}(\mathcal{H}_d)$. 
While relying only on prior information, with no true measurements data available, we use $\mathbf{S}_i(\rho_0)$
to replace the constraints $\bar{\mathbf{S}}_i$ 
in the criteria \textbf{C1} and \textbf{C2}. Note that $\mathbf{S}_i(\rho_0)=\bar{\mathbf{S}}_i$ if the state is perfect generated, i.e., $\mathrm{Tr}(\rho_0A_i)=y_i$.

Obviously, since $\rho_0\in \mathbf{S}_i(\rho_0)$ for all $i\in\{1,\dots,R\}$ by construction, then in this scenario \textbf{C1} can never be satisfied. Thus, we only exploit \textbf{C2} to determine the order of measurements. Suppose that the distance function $D$ is continuous on $\mathcal{S}(\mathcal{H}_d)$, e.g., any matrix norm, quantum relative entropy, etc., (see~\cite[Chapter 9, 11]{nielsen2002quantum} for standard options), due to the compactness of $\bigcap_i\mathbf{S}_i(\rho_0)$, $\max_{\rho\in\bigcap_i\mathbf{S}_i(\rho_0)} D(\rho,\rho_0)$ exists. 

If the state was actually $\rho_0$, the minimal amount of measurements that allow to determine that the preparation was indeed accurate would correspond, according to {\bf C2}, to the minimum $n$ for which there exists a set of measurements indexes $\mathsf{M}_n\subset\{1,\dots,R\}$ such that \begin{equation*}
    \max_{\rho\in \bigcap_{i\in \mathsf{M}_n}\mathbf{S}_i(\rho_0)}D(\rho,\rho_0) \leq \epsilon,\end{equation*} 
 and the optimal sequence would be any permutation of the $\mathsf{M}_n$. 

The Algorithm OS could be used to generate one such optimal sequence.
\begin{table}[h!]
\begin{tabular}{l}
\hline\\
{\normalsize \textbf{Algorithm OS}: Optimal verification Sequence based on $\rho_0$ } 
\\ 
\hline
\vbox{
\begin{itemize}
    \item \textbf{Initialization:}  Define the set $\mathsf{S}=\{1,\dots,R\}$ and set $k=1$.
    \item \textbf{Step 1:} Denote a $k$ elements sequence by $\mathsf{M}_k:=(m_{k,1},\dots,m_{k,k})\in \mathsf{S}^k$. Compute 
    $$
    \overline{\mathsf{M}}_k\in \operatorname*{arg\,min}_{\mathsf{M}_k\in \mathsf{S}^k } \left( \max_{\rho\in \bigcap_{i\in \mathsf{M}_k}\mathbf{S}_i(\rho_0)}D(\rho,\rho_0)\right).
    $$
    If $\max_{\rho\in \bigcap_{i\in \overline{\mathsf{M}}_k}\mathbf{S}_i(\rho_0)}D(\rho,\rho_0)\leq \epsilon$ \textbf{stop the procedure}. Otherwise, 
    update $k=k+1$.
    \item \textbf{Step 2:} Repeat \textbf{Step 1} until $k=d^2$.
\end{itemize}
}\\
\hline
\end{tabular}
\end{table}

Note that each step of above algorithm is independent, thus for some $i<j$, $\overline{\mathsf{M}}_i\not\subset\overline{\mathsf{M}}_j$. At the end of process, we obtain a sequence of measurements $\overline{\mathsf{M}}_n$ containing $n\leq R$ elements, whose corresponding observables are the optimal choice for the verification of $(\epsilon,D,\rho_0)$ accuracy of $\rho_{\rm exp}=\rho_0$. The order of the elements belonging to $\overline{\mathsf{M}}_n$ is not important. 
However, the computational complexity of the above algorithm is too large, in order to determine $\overline{\mathsf{M}}_n$, it needs to solve $\sum^n_{k=1}\left(\begin{smallmatrix} n\\k  \end{smallmatrix}\right)=2^n-1$ optimization problems. 
Moreover, in practice, the generated state $\rho_{\rm exp}$ is usually different from the target state $\rho_0$, thus the generated measurement sequence $\overline{\mathsf{M}}_n$ by Algorithm OS with respect to $\epsilon$ may not be able to verify the accuracy of $\rho_{\rm exp}.$ To obtain a tomographically-complete sequence one needs to add $d^2-n$ linearly-independent measurements operators from the available set.

\subsection{Iterative construction of verification sequences} 

In order to address the above issues, we propose to construct the sequence of measurements iteratively, based on the previous determined measurement indexes, which can greatly reduce the computational complexity and allow to extend the procedure to the full observable set. The resulting sequence will be in general suboptimal with respect to $\rho_0$, but still yields an 
advantage with respect to a random sequence of observables, as shown in Section \ref{sec:simulations}.

\subsubsection{Optimization-based approach}
The general algorithm we propose works as follows:
it starts by evaluating, for each measurement $A_i$, the maximal distance $\alpha_i$ with respect to $\rho_0$ of the states $\rho$ belonging to ${\bf S}_i(\rho_0)$, the set of states that are compatible with the measurement outcome ${\rm Tr}(\rho A_i)={\rm Tr}(\rho_0 A_i)$. The measurement giving the minimum value of $\alpha_i$ is selected as first measurement $A_{m_1}$, and the corresponding maximum distance is $\alpha^1_{m_1}$.
Now, the next measurement $A_{m_{i+1}}$ is chosen so that it is linearly independent on the previously chosen ones and at the same time minimizes the maximum distance of the new compatible set with the measurement of {\it all the previously selected} $A_{m_1},\ldots ,A_{m_i}$. 
The minimum worst-case distance among compatible states $\alpha^n_i,$ with $n$ indicating the iteration and $i$ the selected measurement, is chosen as an indicator of how likely it is that checking \textbf{C2} will allow us to determine whether the actual state is $(\epsilon,D,\rho_0)$-accurate. 

A more formal form of the above algorithm is summarized as Algorithm IOS.
\begin{table}[h!]
\begin{tabular}{l}
\hline\\
{\normalsize \textbf{Algorithm IOS}: Iteratively Optimized Sequence based on $\rho_0$} 
\\ 
\hline
\vbox{
\begin{itemize}
    \item \textbf{Initialization:}  Define the sets $\mathsf{M}=\emptyset$ and $\mathsf{S}=\{1,\dots,R\}.$  Set $k=0$.
    \item \textbf{Step 1:} Define $\bar{\mathsf{S}}$ as the set of all $i\in\mathsf{S}$ such that $A_i\notin \mathrm{span}\{A_n\}_{n\in\mathsf{M}}$. Set $k=k+1$. For all $i\in \bar{\mathsf{S}}$, compute 
$$
\alpha^k_i:=\max_{\rho\in  \bigcap_{n\in\mathsf{M}}\mathbf{S}_n(\rho_0)\cap \mathbf{S}_i(\rho_0) } D(\rho,\rho_0).
$$
If $\min_{i\in \bar{\mathsf{S}}}\alpha^k_i=0$, set $\mathsf{M}=\mathsf{M}\cup \bar{\mathsf{S}}$ and \textbf{stop the process}: in this case $\rho_0$ must belong to the span of the selected measurements.

Otherwise, compute $\operatorname*{arg\,min}_{i\in \bar{\mathsf{S}}}\alpha^k_i$.
If $\arg\min$ output a single integer, set $m_k=\operatorname*{arg\,min}_{i\in \bar{\mathsf{S}}}\alpha^k_i$.
If $\arg\min$ output multiple integers, designate a unique $m_k$ in that set, according to some deterministic rule or at random: in this case the criteria we consider do not lead to a preferred choice.

Update $\mathsf{M}=\mathsf{M}\cup \{m_k\}$, $\mathsf{S}=\mathsf{S}\setminus \{m_k\}$.
\item \textbf{Step 2:} Repeat \textbf{Step 1} until $\mathrm{card}(\mathsf{M})=d^2$
\end{itemize}
}\\
\hline
\end{tabular}
\end{table}



At the end of the procedure, $\mathsf{M}$ is a ordered sequence of measurements, from the most to the less informative based on \textit{a priori} state. 
Note that, at the end of Step 2, we obtain a sequence of measurements containing ${n}$ linearly independent observables, from which the target state $\rho_0$ can be reconstructed via tomography. By construction $\alpha^{1}_{m_1}\geq\alpha^{2}_{m_2}\geq\cdots\geq\alpha^{n}_{m_{n}}$ is a decreasing sequence of the maximum distance from $\rho_0$ of the states compatible with the measurement. However, in practice $\rho_{\rm exp}\neq \rho_0$, for the case of $n<d^2$, the $n$ observables may not be sufficient to verify the accuracy of $\rho_{\rm exp}$. Thus, we need to complete the sequence with additional $d^2-n$ linearly independent observables, which we can choose at random or according to other criteria.

\subsubsection{Analytic approach based on distance bound}
The computational complexity of Algorithm IOS is still highly dependent on the number of optimization problems to be solved that, albeit reduced with respect to the optimal \textit{a priori} sequence, still increases quadratically with the dimension of the Hilbert space. To address this issue, we provide an approximation of Algorithm IOS when the Hilbert-Schmidt distance is chosen as the distance function. In this case, we are not ordering the measurements by evaluating the exact maximal distance of the set of states compatible with the measurement (i.e., the $\alpha^k_i$ values), but instead by evaluating an upper bound of such distance that can be expressed analytically.

The Hilbert-Schmidt distance is defined as
$$
d_{HS}(\rho_0,\sigma):=\sqrt{\mathrm{Tr}(\rho_0-\sigma)^2}, \quad \forall\,\rho_0,\sigma \in \mathcal{S}(\mathcal{H}_d),
$$
In the following proposition, we provide an upper bound of the distance on the target state $\rho_0$ for states $\sigma$ that are compatible with $\rho_0$ according to a set of observables $\{A_{i}\}^K_{i=1}$ where $K\leq R$.

\begin{proposition}
Given a state $\rho_0\in\mathcal{S}(\mathcal{H}_d)$ and 
a set of observables $A_{i}\in\mathcal{B}_{*}(\mathcal{H}_d)$ 
with $i\in\{1,\dots,K\}$, for any $\sigma\in\bigcap^K_{i=1}\mathbf{S}_{i}(\rho_0)$, then
\begin{equation}
    d_{HS}(\rho_0,\sigma) \leq \sqrt{1-{\rm Tr}(\varrho_K^2)}+
    \sqrt{{\rm Tr}(\rho_0^2)-{\rm Tr}(\varrho_K^2)}
    \label{better_constraint}
\end{equation}
where $\varrho_K$ is the projection of $\rho_0$ in the subspace spanned by the operators $\{A_i\}^K_{i=1}$.
\label{Prop:offlineApp}
\end{proposition}
\begin{proof}
The square of the Hilbert-Schmidt distance can be written as $d^2_{HS}(\sigma,\rho_0)={\rm Tr}(\rho_0^2)+{\rm Tr}(\sigma^2)
-2{\rm Tr}(\sigma\rho_0)
\leq 1+{\rm Tr}(\rho_0^2)
-2{\rm Tr}(\sigma\rho_0)$.
Any state $\sigma\in\bigcap_i\mathbf{S}_i(\rho_0)$ satisfies ${\rm Tr}(\sigma A_{i})={\rm Tr}(\rho_0 A_{i})$ for all $i\in\{1,\dots,K\}$.
Therefore, the orthogonal projection of $\rho_0$ and $\sigma$ on the space 
spanned by the operators $\{A_i\}^K_{i=1}$ is the same: we can defined it as $\varrho_K$.
We can thus write $\rho_0=\varrho_K+\varrho_K^\perp$ and $\sigma=\varrho_K+\varsigma_K^\perp$ with $\varrho_K^\perp$ and $\varsigma_K^\perp$ orthogonal to $\varrho_K$ according to the Hilbert-Schmidt inner product, i.e., $\langle \rho,\sigma\rangle_{HS}={\rm Tr}(\rho^*\sigma)$.
Therefore ${\rm Tr}(\rho_0\sigma)={\rm Tr}(\varrho_K^2)+{\rm Tr}(\varrho_K^\perp \varsigma_K^\perp)$. Moreover the equations
${\rm Tr}(\rho_0^2)={\rm Tr}(\varrho_K^2)+{\rm Tr}[(\varrho_K^\perp)^2]$
and 
${\rm Tr}(\sigma^2)={\rm Tr}(\varrho_K^2)+{\rm Tr}[(\varsigma_K^\perp)^2]\leq1$ imply that
${\rm Tr}[(\varrho_K^\perp)^2]={\rm Tr}(\rho_0^2)-{\rm Tr}(\varrho_K^2)$
and
${\rm Tr}[(\varsigma_K^\perp)^2]\leq1-{\rm Tr}(\varrho_K^2)$.
From the Cauchy-Schwarz inequality we have that
${\rm Tr}(\varrho_K^\perp \varsigma_K^\perp)\geq-
\sqrt{{\rm Tr}[(\varrho_K^\perp)^2]{\rm Tr}[(\varsigma_K^\perp)^2]}
\geq
-\sqrt{{\rm Tr}(\rho_0^2)-{\rm Tr}(\varrho_K^2)}\sqrt{1-{\rm Tr}(\varrho_K^2)}$.
We have therefore proved that
\begin{equation}
   {\rm Tr}(\rho_0\sigma)\geq
{\rm Tr}(\varrho_K^2)-\sqrt{{\rm Tr}(\rho_0^2)-{\rm Tr}(\varrho_K^2)}\sqrt{1-{\rm Tr}(\varrho_K^2)} 
\label{Eq:LowBoundTr}
\end{equation}
and the main proposition follows.     
\end{proof} 

\begin{remark}
We would like to point out that if the target state $\rho_0$ is pure (i.e., ${\rm Tr}(\rho_0^2)=1$) the upper bound given in \eqref{better_constraint} simplifies to
\begin{equation}
    \label{better_constraint_pure}
      d_{HS}(\rho_0,\sigma) \leq 2\sqrt{1-{\rm Tr}(\varrho_K^2)}.
\end{equation}
Moreover, a similar bound also holds when the Bures metric is employed and the target state $\rho_0$ is pure. Indeed, when $\rho_0$ is pure, the Bures distance is written as $d_B(\rho_0,\sigma)=\sqrt{2(1-\sqrt{\tr(\rho_0\sigma)}}$. Therefore, by following similar step it possible to demonstrate that, for pure $\rho_0$ for which ${\rm Tr}(\varrho_K^2)\geq 1/2$ we have
\begin{equation}
   d_B^2(\rho_0,\sigma)\leq 
2(1-\sqrt{2{\rm Tr}(\varrho_K^2)-1}). 
\label{better_constraint_Bures}
\end{equation}
\end{remark}

Lastly, the bound~\eqref{better_constraint} can be interpreted geometrically.
The states $\sigma$ are written as $\sigma=\varrho_K+\varsigma_K^\perp$ with fixed $\varrho_K$.
Therefore the states $\sigma$ are contained within a ball centered in $\varrho_K$ and radius $R_K=\max \sqrt{{\rm Tr}[(\varsigma_K^\perp)^2]}=
\sqrt{1-{\rm Tr}(\varrho_K^2)}$. The state $\rho_0=\varrho_K+\varrho_K^\perp$ also belongs to such ball, but its distance from the center is given by $d_K=d_{HS}(\rho_0,\varrho_K)=\sqrt{{\rm Tr}[(\varrho_K^\perp)^2]}=\sqrt{{\rm Tr}(\rho_0^2)-{\rm Tr}(\varrho_K^2)}$.
Therefore, the maximum distance between $\rho_0$ and $\sigma$ is indeed bounded by  $R_K+d_K$, as in \eqref{better_constraint}.
Notice that, by starting from a set of linear independent observables $\{A_i\}$, adding an extra observable $A_j$ will improve the bound.

\begin{proposition}
  Assume we have fixed the first $\{A_i\}_{i=1}^{K}$ and we add a further measurement operator $A_{K+1}$. Let $\{\Gamma_i\}$ be  an orthonormal basis of the space spanned by the $\{A_i\}_{i=1}^{K}$. Define $A^\perp_{K+1}=A_{K+1}-\sum_i{\rm Tr}(\Gamma_iA_{K+1})\Gamma_i.$
Then the projected state becomes: 
\begin{align}
\varrho_{K+1}=
\varrho_K+\frac{
{\rm Tr}(\rho_0A_{K+1}^\perp)}
{
{\rm Tr}[(A^\perp_{K+1})^2]}
A^\perp_{K+1}
\end{align}
The latter also implies $\|\varrho_{K+1}\|^2_{HS}=\|\varrho_{K}\|^2_{HS}+\frac{{\rm Tr}^2(\rho_0 A^\perp_{K+1})}{\|A^\perp_{K+1}\|^2_{HS}}$. 
\end{proposition}

\begin{proof}
We can write $\rho_0=\varrho_K+\alpha A^\perp_{K+1}+\tau^\perp_{K+1},$ with $A^\perp_{K+1}$ orthogonal to all $A_i$'s and
$\tau^\perp_{K+1}$ orthogonal to both $\varrho_K$ and $A^\perp_{K+1}$. 
Since ${\rm Tr}(\rho_0 A_{K+1})={\rm Tr}(\varrho_K A_{K+1})+\alpha{\rm Tr}(A_{K+1} A^\perp_{K+1})$ we can determine $\alpha={\rm Tr}[(\rho_0-\varrho_K)A_{K+1}]/{\rm Tr}(A_{K+1}A_{K+1}^\perp)$.
Therefore the projection of $\rho_0$ into the subspace spanned by the $\{A_i\}$ and $A_{K+1}$
is given by: 
\begin{align}
\varrho_{K+1}&=
\varrho_K+\frac{
{\rm Tr}[(\rho_0-\varrho_K)A_{K+1}]}
{
{\rm Tr}(A_{K+1}^2)-\sum_i
[{\rm Tr}(A_{K+1}\Gamma_i)]^2}
A^\perp_{K+1}.
\end{align}

More in detail, write $\varrho_K:=\sum_n\mathrm{Tr}(\rho_0\Gamma_n)\Gamma_n$ and $A_{K+1}=A^\perp_{K+1}+\Omega_{K+1}$ with $\Omega_{K+1}:=\sum_n\mathrm{Tr}(A_{K+1}\Gamma_n)\Gamma_n$, where $\mathrm{Tr}(\Gamma_n\Gamma_m)=\delta_{n,m}$. We have $\mathrm{Tr}(\varrho_K A^\perp_{K+1})=0$, $\mathrm{Tr}(\rho_0 \Omega_{K+1}) = \sum_n\mathrm{Tr}(A_{K+1}\Gamma_n)\mathrm{Tr}(\rho_0\Gamma_n)$ and 
\begin{equation*}
\begin{split}
  \mathrm{Tr}(\varrho_K \Omega_{K+1}) &= \sum_n\mathrm{Tr}(\rho_0\Gamma_n)\mathrm{Tr}(\Omega_{K+1}\Gamma_n)\\
  &= \sum_{n}\mathrm{Tr}(\rho_0\Gamma_n)\mathrm{Tr}\big(\sum_m \mathrm{Tr}(A_{K+1}\Gamma_m)\Gamma_m \Gamma_n\big) \\
  &= \sum_{n,m}\mathrm{Tr}(\rho_0\Gamma_n)\mathrm{Tr}(A_{K+1}\Gamma_m)\mathrm{Tr}(\Gamma_m \Gamma_n)\\
  &=\sum_n\mathrm{Tr}(\rho_0\Gamma_n)\mathrm{Tr}(A_{K+1} \Gamma_n) = \mathrm{Tr}(\rho_0 \Omega_{K+1}).
  \end{split}
\end{equation*}
From the latter we have that:
\begin{equation*}
\begin{split}
\mathrm{Tr}\big((\rho_0-\varrho_K)A_{K+1}\big) &= \mathrm{Tr}\big((\rho_0-\varrho_K)(A^\perp_{K+1}+\Omega_{K+1})\big) \\&= \mathrm{Tr}(\rho_0 A^\perp_{K+1})-\mathrm{Tr}(\varrho_K A^\perp_{K+1})+\mathrm{Tr}(\rho_0 \Omega_{K+1})-\mathrm{Tr}(\varrho_K \Omega_{K+1}) \\&= \mathrm{Tr}(\rho_0 A^\perp_{K+1}).
  \end{split}
\end{equation*}
Hence:
\begin{align}
\varrho_{K+1}&=
\varrho_K+\frac{
{\rm Tr}[(\rho_0-\varrho_K)A_{K+1}]}
{
{\rm Tr}(A_{K+1}^2)-\sum_i
[{\rm Tr}(A_{K+1}\Gamma_i)]^2}
A^\perp_{K+1}\nonumber=\varrho_K+\frac{{\rm Tr}(\rho_0 A^\perp_{K+1})A^\perp_{K+1}}{\|A^\perp_{K+1}\|^2_{HS}}.
\label{iteration}
\end{align}  
\end{proof}

Notice that the rhs of \eqref{better_constraint} represent an upper bound on the parameter $\alpha^k_i$ defined in Algorithm~IOS.
Since $\|\varrho_K\|_{HS}=\sqrt{{\rm Tr}(\varrho_K^2)}$, according to Proposition~\ref{Prop:offlineApp}, the norm $\|\varrho_K\|_{HS}$ of the projection $\varrho_K$ of $\rho_0$ over the subspace spanned by a subset of observables $\{ A_i\}$ is an useful parameter to optimize the sequence of the measurements. 
The larger is $\|\varrho_K\|_{HS}$, the lower the upper bound on $d_{HS}(\rho_0,\sigma)$.
Therefore, the measurement sequence should be chosen in order to maximize the norm of such projection at each step, since the upper bound~\eqref{better_constraint} is monotonically non-increasing with respect to the norm of the projection. 
To this aim, it is sufficient to select an observable $A_{K+1}$ which maximizes the value of $\frac{{\rm Tr}^2(\rho_0 A^\perp_{K+1})}{\|A^\perp_{K+1}\|^2_{HS}}$ at each step.

A more formal form of the above algorithm is summarized as Algorithm IAS.
\begin{table}[h!]
\begin{tabular}{l}
\hline\\
{\normalsize \textbf{Algorithm IAS}: Iterative Sequence based on $\rho_0$ and the Analytic bound} 
\\ 
\hline
\vbox{
\begin{itemize}
    \item \textbf{Initialization:} Define the sets $\mathsf{M}=\emptyset$ and $\mathsf{S}=\{1,\dots,R\}$. Set $k=1$.
    \item \textbf{Step 1:} For all $j\in\mathsf{S}$, compute 
$A_j^\perp=A_j-\sum_{i\in {\sf M}}{\rm Tr}(A_j\Gamma_i)\Gamma_i$, 
and 
$
\omega^{(k)}_{j}=\frac{{\rm Tr}^2(\rho_0 A^\perp_j)}{\|A^\perp_j\|_{HS}^2} 
$
for all $j\in{\mathsf{S}}$.
Then define the index
$
m_k \in \operatorname*{arg\,max}_{j\in{\mathsf{S}}}\omega^{(k)}_{j},
$
and the matrix
$
\Gamma_{m_k}=\frac{A^\perp_{m_k}}{\|A^\perp_{m_k}\|_{HS}}.
$
Update $\mathsf{M}=\mathsf{M}\cup \{m_k\}$, $\mathsf{S}=\mathsf{S}\setminus \{m_k\}$. Set $k=k+1$. 
\item \textbf{Step 2:} Repeat \textbf{Step 1} until {$\mathrm{card}(\mathsf{M})=d^2$}.
\end{itemize}
}\\
\hline
\end{tabular}
\end{table}

Note that if $\omega^{(k)}_{j}=0$ then $\rho_0\in\mathrm{span}\{\Gamma_{m_1},\dots,\Gamma_{m_{k-1}}\}$. If the $\arg\max$ in the algorithm above produces more than a single index, one is chosen at random in the set.
The sequence is generated by increasing as much as possible in each cycle the value of $\|\varrho_k\|_{HS}$. 
At the end of the procedure, $\mathsf{M}$ corresponds to an ordered sequence of {$d^2$} linearly independent measurement operators based on the upper bound on the distance from $\rho_0$ provided above.

\subsection{Verification algorithm based on the measurement sequence}
Once we obtained the measurement sequence $\mathsf{M}$ using one of the algorithms above, we can perform the Algorithm VM to verify whether the generated state $\rho_{\rm exp}$ is $(\epsilon,D,\rho_0)$-accurate according to \textbf{C1} and \textbf{C2}.

\begin{table}[h!]
\begin{tabular}{l}
\hline\\
{\normalsize \textbf{Algorithm VM}: Verification of the quantum state based on $\mathsf{M}$} 
\\ 
\hline
\vbox{
\begin{itemize}
    \item \textbf{Initialization:} Set $\mathsf{N}=\{m_1\}$ and $k=1$. 
    \item \textbf{Step 1:} Perform the measurements of  $A_{m_k}$ and collect the sampled average output $y_{m_k}$. Compute 
    $$
    \gamma_k:= \min_{\rho\in\bigcap_{n\in\mathsf{N}}\bar{\mathbf{S}}_n}D(\rho,\rho_0),~ \Gamma_k:= \max_{\rho\in\bigcap_{n\in\mathsf{N}}\bar{\mathbf{S}}_n}D(\rho,\rho_0).
    $$
    \begin{itemize}
        \item If $\gamma_k>\epsilon$, then $\rho_{\rm exp}$ is not $(\epsilon,D,\rho_0)$-accurate and \textbf{stop the procedure};
        \item If $\Gamma_k \leq \epsilon$, then $\rho_{\rm exp}$ is $(\epsilon,D,\rho_0)$-accurate and \textbf{stop the procedure};
        \item Otherwise, update $k=k+1$ and $\mathsf{N}=\mathsf{N}\cup\{m_k\}$.
    \end{itemize}
    \item \textbf{Step 2:} Repeat \textbf{Step 1} until $k=d^2$.
\end{itemize}
}\\
\hline
\end{tabular}
\end{table}

\begin{remark}
At the end of the above algorithm, if the procedure ends with $k=d^2$ we can reconstruct the generated state $\rho_{\rm exp}=\sum_{i\in\mathsf{N}}c_i A_i$ where $\{c_i\}_{i\in\mathsf{N}}$ can be computed for example by
\begin{equation}
\left[\begin{matrix}
c_1\\
\vdots\\
c_K
\end{matrix}\right]=
\left[\begin{matrix}
\tr(A_1 A_1) & \dots & \tr(A_1 A_K)\\
\vdots & \ddots & \vdots\\
\tr(A_K A_1) & \dots & \tr(A_K A_K)
\end{matrix}\right]^{-1}
\left[\begin{matrix}
y_1\\
\vdots\\
y_K
\end{matrix}\right].
\label{Eq:Reconstruction}  
\end{equation}   
\end{remark} 

\section{Adaptive quantum state verification} \label{sec:adaptive}
In the previously proposed algorithms, the measurement sequence was determined off-line (i.e., without performing any measurement) by only leveraging the information on the a-priori state $\rho_0$.
Here, we optimize the verification procedure Algorithm IOS and Algorithm IAS by  also exploiting the measurement data at each step in addition to the \textit{a priori} state to determine the next measurement and then verify the state. We call such protocol {\it adaptive verification.} 

For now, suppose that the {\it measurements are perfect}: namely the sampled output averages correspond to the true expected values for the actual state.
We initialize the algorithm as same as in Algorithm IOS, since before perform the measurements, the \textit{a priori} state is the only accessible information. Compute $\alpha^1_i:=\max_{\rho\in\mathbf{S}_i(\rho_0)} D(\rho,\rho_0)$ for all $i\in \{1,\dots,R\}$ and $m_1\in \arg\min_{i\in \{1,\dots,R\}}\alpha^1_i$. 
If $\arg\min$ cannot assign an unique $m_1$, then we consider the following rule: select an observable at random among those indicated by the criterion of Algorithm IAS, namely those maximizing $\mathrm{Tr}^2(\rho_0 A^\perp_{j})/\|A_j^\perp\|^2_{HS}$.
Then, we perform the measurement $A_{m_1}$ and obtain an empirical estimate of $y_{m_1}=\mathrm{Tr}(\rho_{\rm exp}A_{m_1})$. For the sake of simplicity in presenting the algorithm, we shall here assume we actually obtain the exact value $y_{m_1}.$ The case of imperfect estimates can be treated along the same lines. In order to test both criteria \textbf{C1} and \textbf{C2}, we compute 
$$
\omega_1:= \min_{\rho\in\bar{\mathbf{S}}_{m_1}}D(\rho,\rho_0), \qquad 
\Omega_1:= \max_{\rho\in\bar{\mathbf{S}}_{m_1}}D(\rho,\rho_0).
$$
If $\omega_1>\epsilon$, then $\rho_{\rm exp}$ is not $(\epsilon,D,\rho_0)$-accurate; and if $\Omega_1 \leq \epsilon$, then $\rho_{\rm exp}$ is $(\epsilon,D,\rho_0)$-accurate. Otherwise, we determine an estimate of $\rho_{\rm exp}$ based on the measurement data $y_{m_1}$ by $\rho_1=\arg\min_{\rho\in \bar{\mathbf{S}}_{m_1}} f_{\rho_0}(\rho)$, where 
$f_{\rho_0}(\rho)$ is a continuous function such that $\rho_0=\arg\min_{\rho\in \mathcal{S}(\mathcal{H}_d)} f_{\rho_0}(\rho)$, quantifying information distance between $\rho\in \mathcal{S}(\mathcal{H}_d)$ and $\rho_0\in \mathcal{S}(\mathcal{H}_d)$. 
Common choices for $f$ can be the quantum relative entropy~\cite{zorzi2013minimum}, or any distance function on $\mathcal{S}(\mathcal{H}_d)$~\cite[Chapter 9]{nielsen2002quantum}. Strictly convex functions guarantee the uniqueness of the minimum.  
For all $i\in \{1,\dots,R\}\setminus\{m_1\}$, according to the criteria \textbf{C1} and \textbf{C2}, we compute 
$$
\delta^1_i:=\min_{\rho\in \mathbf{S}_{i}(\rho_1)\cap \bar{\mathbf{S}}_{m_1} } D(\rho,\rho_0),~ \quad\Delta^1_i:=\max_{\rho\in\mathbf{S}_{i}(\rho_1)\cap \bar{\mathbf{S}}_{m_1}} D(\rho,\rho_0),
$$
where $\Delta^1_i\geq \delta^1_i\geq 0$ and $\mathbf{S}_i(\rho_1)=\{\rho\in\mathcal{S}(\mathcal{H}_d)|\,\mathrm{Tr}(\rho A_i)=\mathrm{Tr}(\rho_1 A_i)\}$. Notice that the constrained set now is computed  for $\rho_1,$ which depends on the actual measurement outcomes.
Intuitively, the smaller $\epsilon-\delta^1_i$ (resp. $\Delta^1_i-\epsilon$) is, the more likely \textbf{C1} (resp. \textbf{C2}) is verified (see Figure~\ref{Fig:QSV_adaptive}). 

If for some $i$ we have that $\Delta^1_i \geq \delta^1_i>\epsilon$ or $\epsilon>\Delta^1_i \geq \delta^1_i$, it means that choosing the corresponding measurement is expected to bring the compatible set closer to verify criteria \textbf{C1} or \textbf{C2}, respectively. However, if there exists $i$ such that $\delta^1_i = 0$, 
it implies that $\min\{\epsilon-\delta^1_i,\Delta^1_i-\epsilon\}=\epsilon$ and $\rho_0\in \mathbf{S}_{i}(\rho_1)\cap \bar{\mathbf{S}}_{m_1}$, which means that \textbf{C1} cannot yield the conclusion. 
Thus, if $\delta^1_i = 0$ for all $i$, only $\Delta^1_i$ provides the information for the selection of the next measurement.  
Therefore, in order to maximize the possibility of the successful verification, we set 
\begin{equation*}
  m_2\in
  \begin{cases}
  \arg\min_{i\in\mathsf{S}} \Delta^1_i, & \delta^1_i=0,\forall i\\
  \arg\min_{i\in\mathsf{S}} \big\{ \min\{\epsilon-\delta^1_i,\Delta^1_i-\epsilon\} \big\},& \text{else}.
  \end{cases}
\end{equation*}
If $\arg\min$ cannot assign an unique $m_2$, then we can select one by employing the idea of  Algorithm~IAS,  that is to select an observable at random among those which maximize $\mathrm{Tr}^2(\rho_1 {A}^{\perp}_{j})/\|A_j^\perp \|_{HS}^2$.
\begin{figure}[h]
    \centering
    \includegraphics[height=5cm]{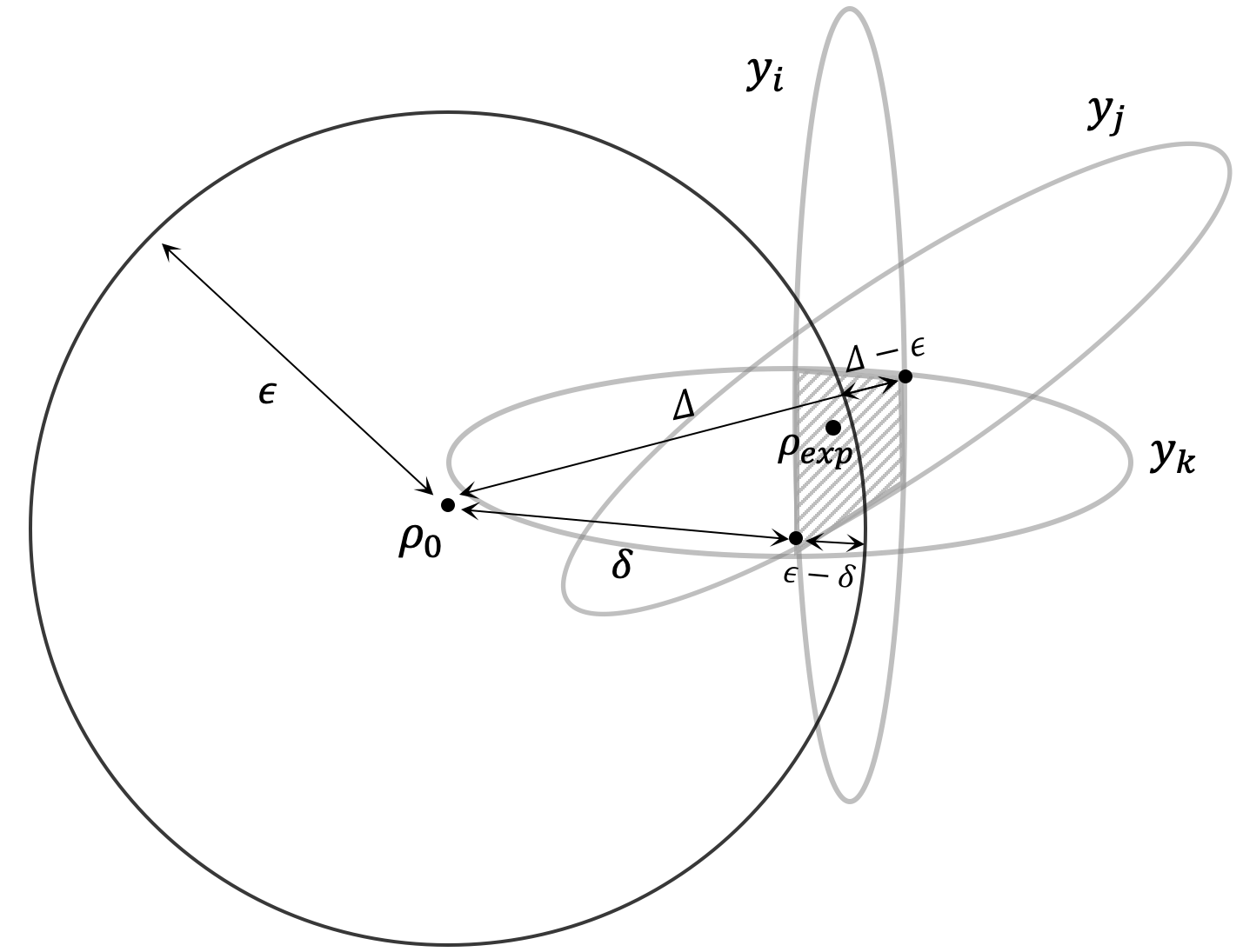} %
    \caption{Diagrams corresponding to the option of quantum state verification criteria \textbf{C1} and \textbf{C2}. The grey area represents $\bar{\mathbf{S}}_i\cap\bar{\mathbf{S}}_j\cap\bar{\mathbf{S}}_k$, i.e., the states compatible with the measurement data $y_i$, $y_j$ and $y_k$.}%
    \label{Fig:QSV_adaptive}
\end{figure}
Then, the whole procedure of verification can be defined recursively. 

\begin{remark}
Note that, at each step, determining an estimate $\rho_k$ of $\rho_{\rm exp}$ solves the quantum state tomography~\cite{paris2004quantum} based on the partial information, the obtained sequence $\{\rho_k\}^R_{k=1}$ converges to $\rho_{\rm exp}$, since $\rho_R=\bigcap^R_{i=1}\bar{\mathbf{S}}_{i}=\rho_{\rm exp}$ and the measurements are supposed to be perfect. 
\end{remark} 

We summarize the algorithm of adaptive verification with perfect measurements as Algorithm AV.
\begin{table}[h!]
\begin{tabular}{l}
\hline\\
{\normalsize \textbf{Algorithm AV}: Adaptive Verification} 
\\ 
\hline
\vbox{
\begin{itemize}
\item \textbf{Initialization:} Define the sets $\mathsf{M}=\emptyset$ and $\mathsf{S}=\{1,\dots,R\}$, and compute $$\alpha^1_i:=\max_{\rho\in\mathbf{S}_i(\rho_0)} D(\rho,\rho_0)$$ for all $i\in \mathsf{S}$ and $\arg\min_{i\in \mathsf{S}}\alpha^1_i$. 

If $\arg\min$ output a single integer, set $m_1=\operatorname*{arg\,min}_{i\in{\mathsf{S}}} \alpha^1_i$. If $\arg\min$ output multiple integers, choose at random a $m_1\in \operatorname*{arg\,min}_{i\in{\mathsf{S}}} \mathrm{Tr}^2(\rho_0 A^\perp_i)/\|A^\perp_i\|^2_{HS}$.
Set $\Gamma_{m_1}=A_{m_1}/\|A_{m_1}\|_{HS}$ and $k=1$.  
\item \textbf{Step 1:} Perform the measurements corresponding to $A_{m_k}$ and collect the sampled average output $y_{m_k}$. Update $\mathsf{M}=\mathsf{M}\cup \{m_k\}$ and $\mathsf{S}=\mathsf{S}\setminus \{m_k\}$. Compute 
\begin{equation*}
    \omega_k:= \min_{\rho\in\bigcap_{n\in\mathsf{M}}\bar{\mathbf{S}}_n}D(\rho,\rho_0),~\Omega_k:= \max_{\rho\in\bigcap_{n\in\mathsf{M}}\bar{\mathbf{S}}_n}D(\rho,\rho_0).
\end{equation*}
    \begin{itemize}
        \item If $\omega_k>\epsilon$, then $\rho_{\rm exp}$ is not $(\epsilon,D,\rho_0)$-accurate and \textbf{stop the procedure};
        \item If $\Omega_k \leq \epsilon$, then $\rho_{\rm exp}$ is $(\epsilon,D,\rho_0)$-accurate and \textbf{stop the procedure};
        \item Otherwise, set 
$
\rho_k\in\operatorname*{arg\,min}_{\rho\in \bigcap_{n\in\mathsf{M}}\bar{\mathbf{S}}_{n}}f_{\rho_0}(\rho).
$
    \end{itemize}
\item \textbf{Step 2:} 
Collect all $i\in\mathsf{S}$ such that $A_i\notin \mathrm{span}\{A_i\}_{i\in\mathsf{M}}$ in $\bar{\mathsf{S}}$.
For all $i\in \bar{\mathsf{S}}$, compute 
\begin{align*}
  \delta^k_i&:=\min_{\rho\in \mathbf{S}_{i}(\rho_{k}) \cap \bigcap_{n\in\mathsf{M}}\bar{\mathbf{S}}_{n} } D(\rho,\rho_0),\\ \Delta^k_i&:=\max_{\rho\in \mathbf{S}_{i}(\rho_{k}) \cap \bigcap_{n\in\mathsf{M}}\bar{\mathbf{S}}_{n}} D(\rho,\rho_0),  
\end{align*}
where $\mathbf{S}_i(\rho_k)=\{\rho\in\mathcal{S}(\mathcal{H}_d)|\,\mathrm{Tr}(\rho A_i)=\mathrm{Tr}(\rho_k A_i)\}$.

If $\delta^k_i=0$ for all $i\in \mathsf{S}$, compute $\operatorname*{arg\,min}_{i\in\bar{\mathsf{S}}} \Delta^{k}_i$.
\begin{itemize}
    \item If $\arg\min$ outputs a single integer, set $m_{k+1}=\operatorname*{arg\,min}_{i\in\bar{\mathsf{S}}} \Delta^{k}_i $;
    \item If $\arg\min$ outputs multiple integers, compute 
$A_j^\perp=A_j-\sum_{i\in {\sf M}}{\rm Tr}(A_j\Gamma_i)\Gamma_i$ for all $j\in\bar{\mathsf{S}}$ and 
choose at random $m_k\in\operatorname*{arg\,max}_{j\in \bar{\mathsf{S}}}\mathrm{Tr}^2(\rho_0 A^\perp_j)/\|A^\perp_j\|^2_{HS}.$ Set $\Gamma_{m_k}=A_{m_k}^\perp/\|A_{m_k}^\perp\|_{HS}.$
\end{itemize}
Otherwise, compute 
$$
\operatorname*{arg\,min}_{i\in\bar{\mathsf{S}}} \big\{ \min\{\epsilon-\delta^{k}_i,\Delta^{k}_i-\epsilon\} \big\}.
$$
\begin{itemize}
    \item If $\arg\min$ outputs a single integer, set $m_{k+1}=\operatorname*{arg\,min}_{i\in\bar{\mathsf{S}}} \big\{ \min\{\epsilon-\delta^{k}_i,\Delta^{k}_i-\epsilon\} \big\}$;
    \item If $\arg\min$ outputs multiple integers, compute 
$A_j^\perp=A_j-\sum_{i\in {\sf M}}{\rm Tr}(A_j\Gamma_i)\Gamma_i$ for all $j\in\bar{\mathsf{S}}$ and 
choose at random $m_k\in\operatorname*{arg\,max}_{j\in \bar{\mathsf{S}}}\mathrm{Tr}^2(\rho_0 A^\perp_j)/\|A^\perp_j\|^2_{HS}.$ Set $\Gamma_{m_k}=A_{m_k}^\perp/\|A_{m_k}^\perp\|_{HS}.$
\end{itemize}
Update $k=k+1$.
\item \textbf{Step 3:} Repeat \textbf{Step 1} and \textbf{Step 2} until $\mathrm{card}(\mathsf{M})=d^2$.
\end{itemize}
}\\
\hline
\end{tabular}
\end{table}
Due to the perfect measurements,  we can always obtain the verification results when the above algorithm ends.
In Step 2, we specifically consider the case $\delta^k_i=0$ for all $i\in \mathsf{S}$, in which $\rho_0$ belongs to the compatible sets, i.e., \textbf{C1} is always verified. Thus, we can only apply \textbf{C2} to determine the next measurement. 
If $\rho_{\rm exp}=\rho_0$, in Step 1 of Algorithm 4, we have $\rho_k \equiv \rho_0$ for any $k\in\{1,\dots,R\}$ since $\rho_0=\arg\min_{\rho\in \mathcal{S}(\mathcal{H}_d)} f_{\rho_0}(\rho)$, which implies $\delta^k_i\equiv 0$. 
Thus, in this case, Algorithm~4 is equivalent to the combination of Algorithm~IOS and Algorithm~IAS. 

Note that Algorithm~AV can also be applied to the imperfect measurement case. However, if the sample size is not big enough or there are errors and bias, one may obtain $\bigcap^K_{i=1}\hat{\mathbf{S}}_{m_i}=\emptyset.$ In this case, we need to stop the verification process and re-measure $\rho_{\rm exp}$.

\section{Application: Two-qubit systems}\label{sec:simulations}
In the following, we test the proposed algorithm simulating measurements to verify the accuracy of preparation of randomized pure states in a two qubit system. We summarize the key elements of the numerical experiments we ran.

\textbf{Target states}: According to the normal distribution, we pick 100 sets of 4 independent complex random numbers with real and imaginary parts belonging to $[-100,100]$, i.e., $|\psi_i\rangle\in\mathbb{C}^4$ with $i=1,\dots,100$. Then, we generate 100 pure target states by $\rho_{0,i}=\frac{|\psi_i\rangle\langle\psi_i|}{\mathrm{Tr}(|\psi_i\rangle\langle\psi_i|)}$.\\
\textbf{Bures distance}: The distance we employ is the Bures diatance, which reduces to $d_B(\rho,\rho_0)=\sqrt{2(1-\sqrt{F(\rho,\rho_0)})}=\sqrt{2(1-\sqrt{\mathrm{Tr}(\rho \rho_0)})}$ for the case of $\rho_0$ being a pure state. Obviously, $d_B(\rho,\rho_0)$ is strictly monotonically decreasing with respect to $\mathrm{Tr}(\rho \rho_0)$. Due to the linearity, we can apply the convex optimization (CVX-SDP~\cite{cvx}) in the simulation for searching the minimum and maximum value of $\mathrm{Tr}(\rho \rho_0)$ under constraints.\\
\textbf{Accuracy}: $\epsilon = \sqrt{2(1-\sqrt{\tilde{\epsilon}})}$, where $\tilde{\epsilon}$ is the desired  precision for the fidelity $\mathrm{Tr}(\rho \rho_0)$. We consider 
$\tilde{\epsilon}=0.95$ so that  $\epsilon= 0.2250$.\\
\textbf{Measurements}: We apply projective measurements into Pauli eigenstates. Let $\Pi_1\dots\Pi_6$ be the eigenprojectors of Pauli matrices corresponding to the eigenvalue $1$ and $-1$ respectively, i.e., $\sigma_x \Pi_1 = \Pi_1, \sigma_x \Pi_2 = -\Pi_2, \dots, \sigma_z \Pi_6 = -\Pi_6$. We denote by $A_{6(i-1)+j}=\Pi_i\otimes\Pi_j$ with $i,j\in\{1,\dots,6\}$ the 36 observables for the two-qubit system. The set of observables $\{A_i\}^{36}_{i=1}$ is information-complete.\\
\textbf{Generated state}: 
We generate 100 full rank $(\epsilon,d_B,\rho_{0,k})$-accurate states $\rho^{a}_{{\rm exp},k}$ and 100 full rank $(\epsilon,d_B,\rho_{0,k})$-non-accurate states $\rho^{n}_{{\rm exp},k}$ by perturbing the target state $\rho_{0,k}$ with $k\in{1,\dots,100}$ as 
\begin{equation}
    \rho_{{\rm exp},k} = e^{i \eta H_k}\left((1-\lambda)\rho_{0,k}+ \frac{\lambda}{4} \mathds{1}_4\right) e^{-i \eta H_k},
\end{equation}
where $\lambda\in(0,1)$, $\eta>0$ and $H_k$ are random Hermitian matrix. We generate the random $H_k\in\mathcal{B}_*(\mathbb{C}^4)$ in the following way, express $H_k = \sum^{15}_{j=0}h_{j,k}\Gamma_n$ where $\Gamma_0 = \mathds{1}_4$ and $\{\Gamma_j\}^{15}_{j=1}$ are generators of the Lie algebra $\mathrm{SU}(4)$ satisfying $\mathrm{Tr}(\Gamma_j)=0$ and $\mathrm{Tr}(\Gamma_m \Gamma_j)=2\delta_{jm}$ with $j,m\in\{1,\dots,15\}$, $\{h_{j,k}\}^{16}_{n=0}$ are random scalars drawn from the uniform distribution in the intervals $(-1,1)$.
We set $\eta = 0.1$, $\lambda = 0.0001$ for the accurate case and $\lambda = 0.1$ for the non-accurate case.

\subsection{Before measurements: Algorithm IOS vs Algorithm IAS}

\textbf{Algorithm IOS}: We use CVX-SDP mode to apply semidefinite programming, and obtain $100$ measurement sequences, $\mathsf{M}_k=[m_{k,j}]_{ j\leq 16}$ for $k\in\{1,\dots,100\}$.\\
\textbf{Algorithm IAS}: We obtain $100$ measurement sequences, $\mathsf{R}_k=[r_{k,j}]_{j\leq 16}$ for $k\in\{1,\dots,100\}$.\\
\textbf{Comparison}: Based on the measurement sequences $\mathsf{R}_k$ generated by Algorithm IOS, we apply semidefinite programming (CVX-SDP mode) to compute the following 
$$
\beta_{k,l} = \max_{\rho\in\bigcap_{j\in [\mathsf{R}_k]_l }\mathbf{S}_{j}(\rho_{0,k})} d_B(\rho,\rho_{0,k}),
$$
where $[\mathsf{R}_k]_l$ denotes the first $l$ elements of $\mathsf{R}_k$. 
The value $\beta_{k}$ can be considered as an indicator of how well Algorithm~IAS approximates Algorithm~IOS. The upper diagram of Figure~\ref{Fig:Ideal_Algo1vsAlgo2} draws error bars of $\beta_{k}-\alpha_{k}$ which represents the mean value and standard deviation, where $\alpha_{k}$ are defined in Algorithm~IOS; the lower diagram draws the number of measurements required by Algorithm IAS minus the one required by Algorithm IOS for reconstructing $\rho_{0,k}$. Taking the machine precision into account, reconstruction of $\rho_{0,k}$ means $d_B(\rho,\rho_{0,k})\leq 10^{-6}$ for $\rho$ belonging to the compatible set. For 100 target states $\rho_{0,k}$, the mean values and the standard deviations of the number of measurements required by Algorithm IOS and Algorithm IAS for the reconstruction are $(5.69,0.5449)$ and $(6.47,0.6884)$ respectively.
 It is worth noting that, more measurements are required by Algorithm IOS than Algorithm IAS in few cases, since Algorithm~IOS does not always provide the optimal measurement sequence, being itself an approximation of Algorithm OS.
\begin{figure}[!t]
    \centering
    \includegraphics[height=8cm]{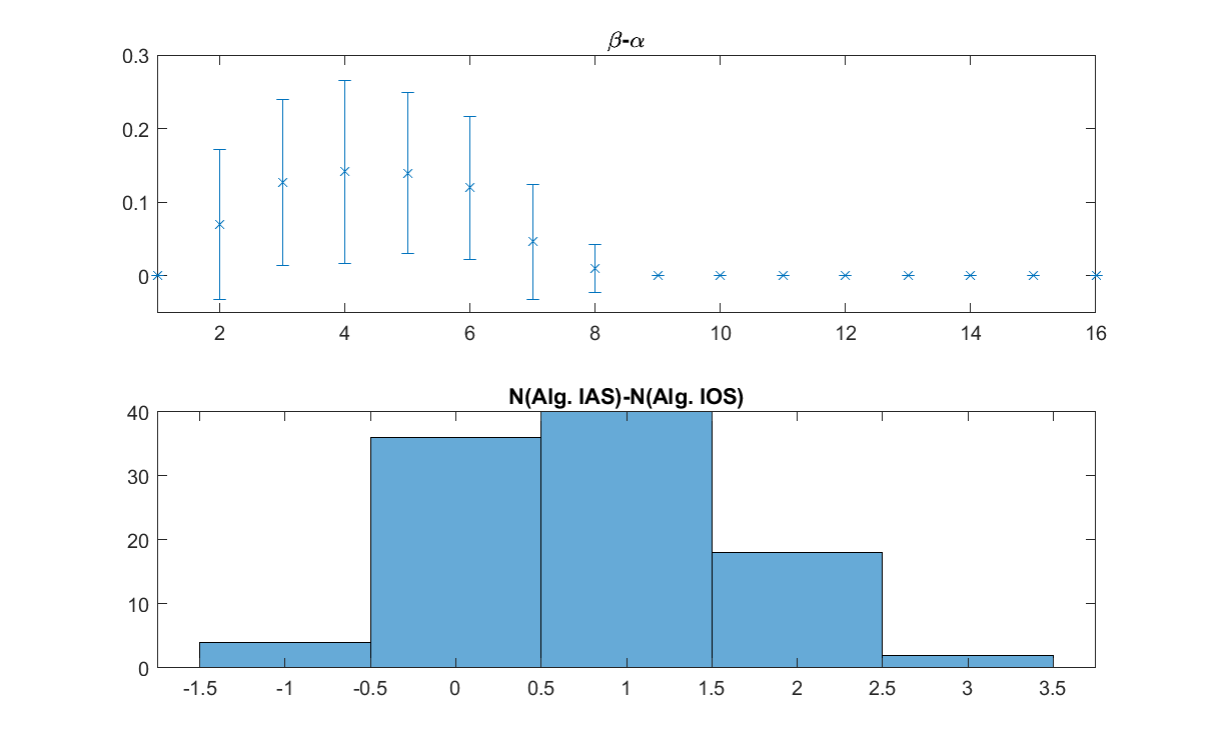}
    \caption{Comparison of Algorithm~IOS and Algorithm~IAS on the reconstruction of $\rho_{0,k}$ with $k\in\{1,\dots,100\}$.}%
    \label{Fig:Ideal_Algo1vsAlgo2}
\end{figure}

\subsection{Accurate $\rho_{\rm exp}$: Algorithm IOS vs Algorithm IAS vs Algorithm AV vs Control groups}

\textbf{Control groups}: Since the set of measurements considered here is information-overcomplete, we generate 5 random measurement sequences for each accurate generated state $\rho^{a}_{{\rm exp},k}$, every sequence contains 16 linearly independent observables. \\ 
\textbf{Numerical Test}: 
We apply the verification protocol (Algorithm VM) on the measurement sequences generated off-line by Algorithm IOS, Algorithm IAS and randomized control groups, and run the adaptive protocol (Algorithm AV) with $f_{\rho_0}(\rho)=d_B(\rho,\rho_0)$.

\begin{remark}
In the case of multiple measurements with the same index of merit, Algorithm IAS selects one measurement at random, while Algorithm IOS and AV use the following rule, inspried by Algorithm IAS: select an observable at random among those which maximize $\mathrm{Tr}^2(\rho_1 {A}^{\perp}_{j})/\|A_j^\perp \|_{HS}^2$. The further optimization step is based on analytic formulas so it is not computationally intensive. The same rule will be used in the next set of simulations as well. 
\end{remark}

The main results are summarized in Figure~\ref{Fig:eps1Acc} and Table~\ref{Tab:NumVer1_acc}. The first diagram of Figure~\ref{Fig:eps1Acc} shows the histogram of the number of measurements required for the verification of accuracy by Algorithm~IOS, Algorithm IAS, Algorithm AV and control groups. This diagram and Table~\ref{Tab:NumVer1_acc} confirm the efficiency of our algorithms in verification of accuracy. The rest diagrams show the histogram of difference of the number of measurements required by different algorithms. 
In this situation, Algorithm~IOS exhibits an advantage with respect 
to Algorithm~IAS. In this case, the performance of Algorithm AV is almost equal to Algorithm~IOS. This results are not surprising: when the state to be verified is indeed close to the target one, Algorithm IOS is expected to provide the best iteratively-built sequence. Nonetheless, Algorithm IAS performance is fairly close (one extra measurement operator is needed on average), and has the advantage of avoiding iterated optimization procedures as it relies only on analytic formulas. 

\begin{remark}
It is worth noticing that  the performance of Algorithm AV strongly depends on the choice of the function $f_{\rho_0}$. Here, we only consider the basic choice $f_{\rho_0}(\rho)=d_B(\rho,\rho_0)$: the optimization of $f_{\rho_0}$ will be the focus of the future work.  
\end{remark}

\begin{figure}[!t]
    \centering
    \includegraphics[height=8cm]{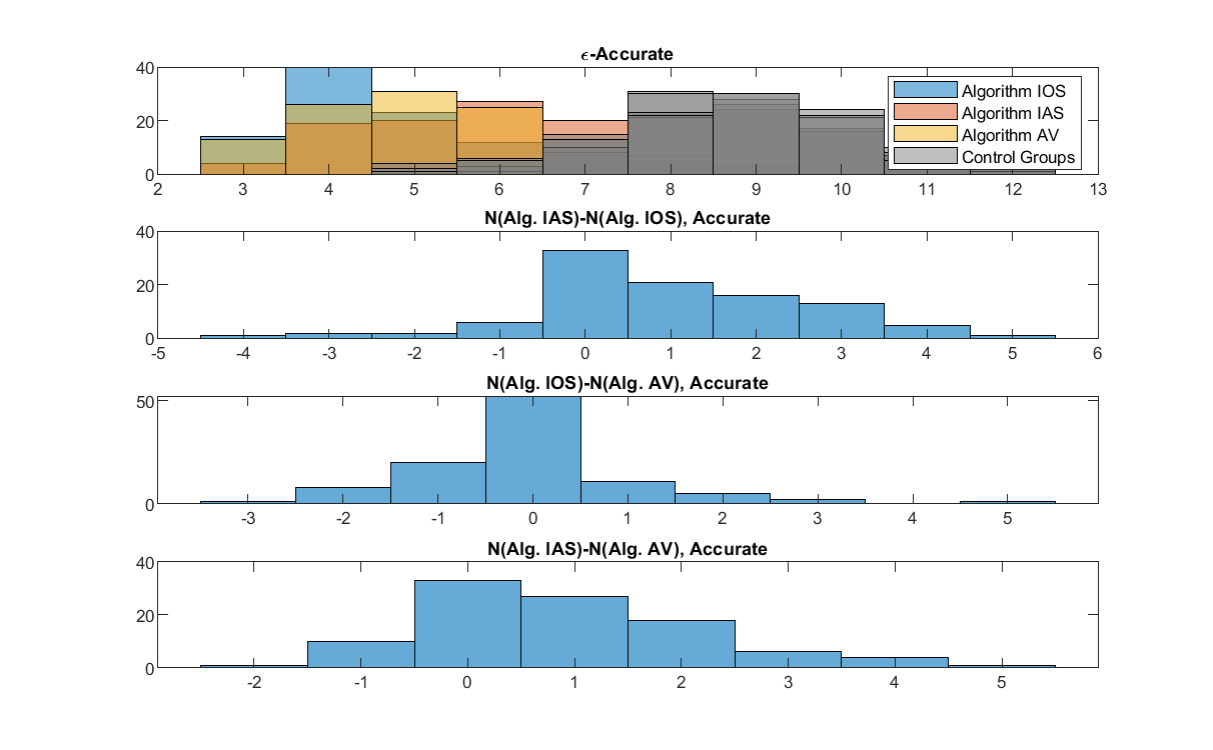}
    \caption{The first histogram displays the distribution of the number of measurements required to verify the accuracy of $\rho^{a}_{\rm{exp},k}$ for $k=1,\dots,100$. The other three show the distribution of the difference between the lengths of the sequences of two algorithms for the same set of generated measurements: For example, if the displayed $N(\textrm{(Alg. X)}-N(\textrm{(Alg. Y)})$ is negative it indicates an advantage for $\textrm{(Alg. X)}.$}
    \label{Fig:eps1Acc}
\end{figure}

\subsection{Non-accurate $\rho_{\rm exp}$: Algorithm IOS vs Algorithm IAS vs Algorithm AV vs Control groups}

\textbf{Control groups}: We generate 5 random measurement sequences for each non-accurate generated state $\rho^{n}_{{\rm exp},k}$, every sequence contains 16 linearly independent observables. \\ 
\textbf{Numerical Test}: We apply the verification protocol (Algorithm  VM) on the measurement sequences generated off-line by Algorithm IOS, Algorithm IAS and randomized control groups, and also run the adaptive protocol (Algorithm AV) with $f_{\rho_0}(\rho)=d_B(\rho,\rho_0)$. 

The main results are summarized in Figure~\ref{Fig:eps1Nacc} and Table~\ref{Tab:NumVer1_nacc}. The first diagram of Figure~\ref{Fig:eps1Nacc} shows the histogram of the number of measurements required for the verification of non-accuracy by Algorithm IOS, Algorithm IAS, Algorithm AV and control groups. This diagram and Table~\ref{Tab:NumVer1_nacc} confirm the efficiency of our algorithms in verification of non-accuracy with respect to random sequences. The rest diagrams show the histogram of difference of the number of measurements required by different algorithms. We can observe that the performance are similar, with a slight advantage for the adaptive protocol, Algorithm AV. Other numerical experiments indicate that the difference in performance becomes more relevant if the needed number of measurements grows.
\begin{sidewaystable}
\caption{Verification of accurate state}\label{Tab:NumVer1_acc}
\begin{tabular}{c|c|c|c|c|c|c|c}
\hline
Alg. IOS & Alg. IAS & Alg. AV & Control gr. 1 & Control gr. 2 & Control gr. 3 & Control gr. 4 & Control gr. 5 \\  \hline
(4.76,1.46) & (5.73,1.42) & (4.83,1.10) & (8.68,1.38) & (8.74,1.52) & (8.82, 1.38) & (8.75,1.38) & (8.71,1.37) \\  \hline
\end{tabular}
\footnotetext{The mean value and the standard deviation $
(m,\sigma)$ of the number of measurements required for verifying the accuracy of $\rho^{a}_{\rm{exp},k}$ for $k=1,\dots,100$.}

\bigskip

\caption{Verification of non-accuration}\label{Tab:NumVer1_nacc}
\begin{tabular}{c|c|c|c|c|c|c|c}
\hline
Alg. IOS & Alg. IAS & Alg. AV & Contr. gr. 1 & Contr. gr. 2 & Contr. gr. 3 & Contr. gr. 4 & Contr. gr. 5 \\  \hline
(5.14,1.65) & (5.28,1.25) & (5.06,1.11) & (8.38,1.80) & (8.34,1.72) & (8.71,1.73) & (8.48, 1.8504) & (8.57,1.83) \\  \hline
\end{tabular}
\footnotetext{The mean value and the standard deviation of the number of measurements required for verifying the non-accuracy of $\rho^{n}_{\rm{exp},k}$ for $k=1,\dots,100$.}
\end{sidewaystable}
\begin{figure}[!t]
    \centering
    \includegraphics[height=8cm]{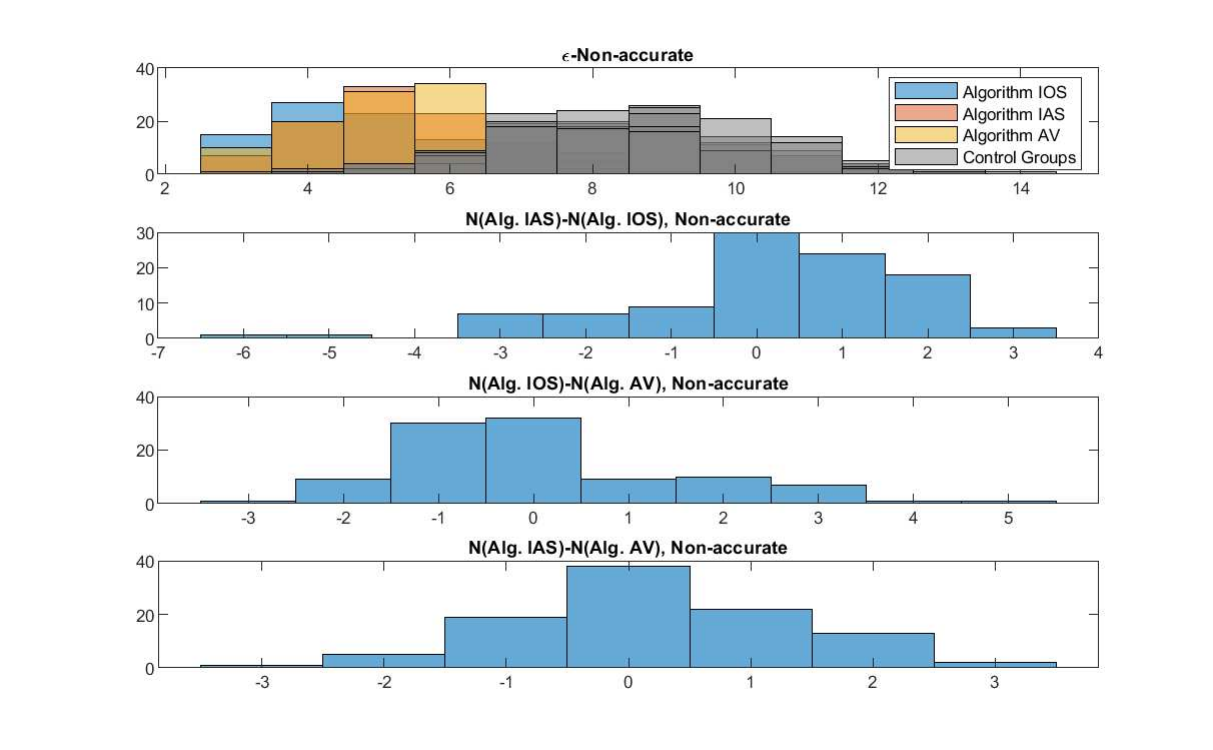} %
    \caption{The first histogram displays the distribution of the number of measurements required to verify the non-accuracy of $\rho^{n,1}_{\rm{exp},k}$ for $k=1,\dots,100$. The other three show the distribution of the difference between the lengths of the sequences of two algorithms for the same set of generated measurements.}
    \label{Fig:eps1Nacc}
\end{figure}

\section{Conclusions}
In this work we define and study {\it quantum state verification}, a key task to test the effectiveness of quantum state preparation procedures, quantum communication channels, quantum memories, and a variety of quantum control algorithms.

Assuming that i.i.d. copies of the system can be produced, the resulting state can be identified by tomographic techniques: sampled averages of a basis of observables are sufficient to determine an estimate of the state and thus to decide if it is compatible with given accuracy requirements. We propose improved strategies to select the observables to be measured, so that a decision on the accuracy of the preparation can be reached well before the full set of measurement is completed. While The protocols rely on prior information about the target state, and either provide a full ordered list of observables to be performed, or adaptively decide the next observable based on the previously obtained ones. While in our approach scales as a linear function of $1/\epsilon^2$ (by applying Proposition 10 of \cite{certificationtutorial} to each measurement of the sequence in order to obtain an appropriate accuracy), all strategies obtaining $1/\epsilon$ scaling rely on the ability of tuning the measurements for the specific target. Here, on the other hand, we are limited to a fixed, finite set of general measurements, a situation motivated by typical  experimentally-available setups.

Numerical tests indicate that, for example, a fidelity of 0.95 can be tested on a qubit system
with just $5$ measurement of joint Pauli operators, when using randomized sequences requires at least $8$. While the solution of the problem leads to solve and compare multiple optimization problems, we also propose an iterative, suboptimal algorithm whose solution can be computed analytically, based on a geometric approximation of the set of states compatible with given measurement outcomes. The adaptive strategy holds an advantage, especially when the needed number of measurements grows, albeit it requires a more involved implementation.
Further work will address the use of optimized measurement sequences for fast tomography, the use of different distance functions for the adaptive strategies, and the application to real data from experimental systems of interest.

\bmhead{Acknowledgments}
F.T. and G.V. acknowledge partial funding from the European Union - NextGenerationEU, within the National Center for HPC, Big Data and Quantum Computing (Project No. CN00000013, CN 1, Spoke 10) and from the European Union’s Horizon Europe research and innovation programme under the project ``Quantum Secure Networks Partnership" (QSNP, grant agreement No. 101114043). Views and opinions expressed are however those of the author(s) only and do not necessarily reflect those of the European Union or European Commission-EU. Neither the European Union nor the granting authority can be held responsible for them.

\section*{Declarations}
The authors have no relevant financial or non-financial interests to disclose. The datasets generated during and/or analyzed during the current study are available from the corresponding author on reasonable request.



\bibliography{Ref_Verification}

\end{document}